\documentclass[runningheads]{llncs}
\usepackage[pdftex]{graphicx}
\usepackage[linesnumbered,ruled,vlined]{algorithm2e}
\usepackage{amsmath}
\usepackage{amssymb}
\usepackage{amsfonts}

\usepackage{amsthm}
\usepackage[capitalise]{cleveref}

\usepackage{stmaryrd}
\usepackage{url}
\usepackage[labelformat=simple]{subcaption}

\newcommand{\set}[1]{\left\{ #1 \right\}}
\newcommand{\setST}[2]{\left\{ #1 \;\middle|\; #2 \right\}}

\newcommand{\edge}[2]{\set{#1, #2}}

\newcommand{\cliqueWidth}[1]{\mathrm{cw}(#1)}
\newcommand{\disUni}{\oplus}
\newcommand{\join}[2]{\eta_{#1, #2}}
\newcommand{\relabel}[2]{\rho_{#1 \rightarrow #2}}

\newcommand{\compl}[1]{\overline{#1}}
\newcommand{\inv}[1]{{#1^I}}
\newcommand{\invCompl}[1]{\compl{\inv{#1}}}

\newcommand{\splitCliqueWidth}[1]{\mathrm{scw}(#1)}

\newcommand{\expr}[2]{\epsilon[#1, #2]}
\newcommand{\splitExpr}[2]{\sigma[#1, #2]}

\begin{document}
\title{On the Clique-Width of Unigraphs}
%
%
\author{Yu Nakahata\inst{1}\orcidID{0000-0002-8947-0994}}
\authorrunning{Yu Nakahata}
%
\institute{Nara Institute of Science and Technology, Takayama 8916-5, Ikoma Nara 6300192, Japan
\email{yu.nakahata@is.naist.jp}
}
\maketitle              
\begin{abstract}
Clique-width is a well-studied graph parameter.
For graphs of bounded clique-width, many problems that are NP-hard in general can be polynomial-time solvable.
The fact motivates several studies to investigate whether the clique-width of graphs in a certain class is bounded or not.
We focus on unigraphs, that is, graphs that are uniquely determined by their degree sequences up to isomorphism.
We show that every unigraph has clique-width at most 4.
It follows that many problems that are NP-hard in general are polynomial-time solvable for unigraphs.
\end{abstract}

\section{Introduction}
\emph{Clique-width} is a well-studied graph parameter~\cite{courcelle1993handle}.
Clique-width can be seen as a generalization of another well-known graph parameter, \emph{treewidth}.
If the treewidth of a graph is a constant, its clique-width is a constant~\cite{corneil2005relationship}.
The converse is not true in general.
For example, the complete graph with $n$ vertices has the treewidth $n$ but the clique-width 2 regardless of $n$.
As shown in this example, clique-width can be bounded by a constant even for dense graphs, unlike treewidth.
If the clique-width of a graph class is bounded, many problems that are NP-hard in general can be polynomial-time solvable for the class~\cite{courcelle2000linear,kobler2003edge,rao2007msol}.
The fact motivates several studies to show that the clique-width of some graph classes are bounded~\cite{courcelle2000upper,golumbic2000clique} and that of others are not~\cite{golumbic2000clique,brandstadt2003linear}.

As a graph class whose clique-width is not known to be bounded, we focus on \emph{unigraphs}~\cite{johnson1975simple,li1975graphic}, that is, graphs uniquely determined by their degree sequences up to isomorphism.
Unigraphs include important graph classes such as threshold graphs~\cite{chvatal1977aggregation}, split matrogenic graphs~\cite{hammer2004splitoids}, matroidal graphs~\cite{peled1977matroidal}, and matrogenic graphs~\cite{foldes1978class}.
For all these subclasses, it is known that the clique-width is at most 4~\cite{courcelle2000linear}.
However, it is open for unigraphs.
We think that there are two main reasons for the difference.
First, although all the subclasses are hereditary, that is, closed under taking induced subgraphs, unigraphs are not as shown in \cref{fig:unigraph_not_hereditary}.
It makes the analysis of unigraphs difficult.
Second, although there are many graph-theoretic studies for unigraphs, there are few algorithmic ones.
Analyzing the clique-width of unigraphs is important to reveal algorithmic aspects of unigraphs.

\begin{figure}[t]
    \centering
    \begin{subfigure}{.30\linewidth}
        \centering
        \includegraphics[bb=0 0 148 123, scale=0.5]{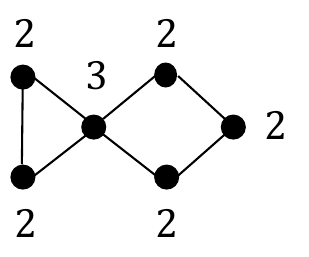}
        \caption{Unigraph $G$.}
        \label{fig:unigraph}
    \end{subfigure}
    \hfill
    \begin{subfigure}{.30\linewidth}
        \centering
        \includegraphics[bb=0 0 149 123, scale=0.5]{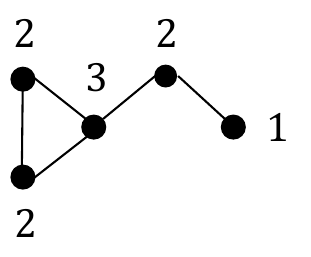}
        \caption{Induced subgraph $F$ of $G$.}
        \label{fig:unigraph_subgraph1}
    \end{subfigure}
    \hfill
    \begin{subfigure}{.30\linewidth}
        \centering
        \includegraphics[bb=0 0 148 123, scale=0.5]{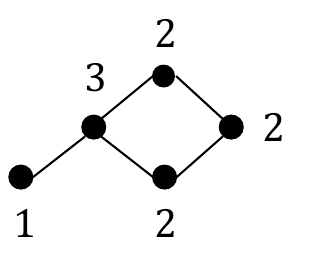}
        \caption{Graph that has the same degree sequence $(3, 2, 2, 2, 1)$ as $F$ but is not isomorphic to $F$.}
        \label{fig:unigraph_subgraph2}
    \end{subfigure}
    \caption{Examples to show that unigraphs are not hereditary. The number beside a vertex is its degree.}
    \label{fig:unigraph_not_hereditary}
\end{figure}

In this paper, we show that the clique-width of unigraphs is at most 4.
It follows that many problems that are NP-hard in general are polynomial-time solvable for unigraphs.
We prove our result by a relationship between clique-width and the characterization of unigraphs based on the canonical decomposition of graphs given by Tyshkevich~\cite{tyshkevich2000decomposition}.

This paper is organized as follows.
The rest of this section summarizes related work.
\cref{sec:preliminaries} gives preliminaries on graphs (\cref{sec:graphs}), clique-width (\cref{sec:clique-width}), and the canonical decomposition of graphs and the characterization of unigraphs (\cref{sec:unigraph}).
In \cref{sec:main}, we show our main result: the clique-width of unigraphs is at most 4.

\subsubsection*{Related Work.}
Clique-width is introduced by Courcelle et al.~\cite{courcelle1993handle}.
Although calculating the clique-width of a graph is NP-hard in general~\cite{fellows2009clique}, whether the clique-width of a graph is at most 3 or not can be determined in polynomial time~\cite{corneil2012polynomial}.
By Courcelle's theorem, every problem can be written in so-called MSO$_1$ admits a linear-time algorithm for graphs of bounded clique-width~\cite{courcelle2000linear}.
These problems include \textsc{Clique}, \textsc{Vertex Cover}, and \textsc{Dominating Set}.
Our result implies that all such problems are linear-time solvable for unigraphs.
Clique-width is related to other graph parameters than treewidth.
For example, clique-width is constant if and only if rank-width~\cite{oum2006approximating} or NLC-width~\cite{johansson1998clique} is constant.
Our result indicates that both the parameters are bounded for unigraphs.
For some graph classes, whether the clique-width of the class is bounded by a constant or not is studied.
For example, cographs are exactly the graphs with clique-width at most 2~\cite{courcelle2000upper} and the clique-width of distance-hereditary graphs is at most 3~\cite{golumbic2000clique}.
In contrast, the clique-width is unbounded for unit interval graphs~\cite{golumbic2000clique} and bipartite permutation graphs~\cite{brandstadt2003linear}.

Unigraphs and its subclasses are well-studied~\cite{johnson1975simple,li1975graphic,hammer2004splitoids,chvatal1977aggregation,peled1977matroidal,foldes1978class}.
Although they are intensively studied from a graph-theoretic view,
there are a few studies from an algorithmic view.
As a few examples, there are linear-time recognition algorithms for unigraphs~\cite{borri2011recognition,kleitman1975note}.
Calamoneri and Petreschi~\cite{calamoneri2011labeling} give an approximation algorithm for $L(2, 1)$-labeling, a variant of graph coloring, of unigraphs.
However, it is open whether the problem is NP-hard for unigraphs.
Since, $L(2, 1)$-labeling can be written in MSO$_1$~\cite{calamoneri2011labeling},
our result implies that we can decide whether there exists an $L(2, 1)$-labeling using $\ell$ colors for a fixed value of $\ell$ in linear time for unigraphs.
In contrast, it is still open whether $L(2, 1)$-labeling for not fixed $\ell$ can be polynomial-time solvable for unigraphs.

Graphs with few $P_4$'s, i.e., $P_4$-sparse graphs and partner-limited graphs, are known to have clique-width at most 4~\cite{kaminski2009recent}.
Unigraphs are not contained in these classes since the double star $S_2(p,2)$ (see \cref{sec:unigraph}) is a counterexample.

\section{Preliminaries}\label{sec:preliminaries}
\subsection{Graphs}\label{sec:graphs}
Let $G$ be a graph.
We assume that $G$ is connected and simple (without self-loops and multi-edges).
$V(G)$ and $E(G)$ denotes the vertex and edge sets of $G$, respectively.
The subgraph induced by $V' \subseteq V(G)$ is denoted by $G[V']$.
If $G[V']$ is a complete graph, $V'$ is a \emph{clique}.
If $G[V']$ has no edges, $V'$ is an \emph{independent set}.
$\compl{G}$ denotes the complement graph of $G$.
$K_n$, $P_n$, and $C_n$ denotes the complete, path, and cycle graph with $n$ vertices, respectively.
$K_{n, m}$ denotes the complete bipartite graph with the two parts of $n$ and $m$ vertices.
Especially, $K_{1, n}$ is a \emph{star}; its \emph{center} is the vertex with degree $n$ (when $n=1$, we choose an arbitrary vertex of $K_{1,1}$), and its \emph{leaves} are the other vertices.
For graphs $G$ and $H$ with $V(G) \cap V(H) = \emptyset$, their \emph{disjoint union} is the graph $G \disUni H = (V(G) \cup V(H), E(G) \cup E(H))$.
We define $mK_2$ as the disjoint union of $m$ copies of $K_2$.
For a positive integer $k$, we define $[k] = \set{1, \dots, k}$.

A graph $G$ is a \emph{split graph} if $V(G)$ can be partitioned into two sets $A$ and $B$ such that $A$ is a clique and $B$ is an independent set.
A graph $G$ is a \emph{unigraph} if it is uniquely determined by its degree sequence up to isomorphism.
A graph class $\mathcal{G}$ is \emph{hereditary} if, for every graph $G \in \mathcal{G}$, all induced subgraphs of $G$ are also in $\mathcal{G}$.
Unigraphs are not hereditary.
For example, although the graph in \cref{fig:unigraph} is a unigraph, there is an induced subgraph that is not a unigraph (\cref{fig:unigraph_subgraph1,fig:unigraph_subgraph2}).

\subsection{Clique-Width}\label{sec:clique-width}

\begin{definition}[Clique-width]\label{def:clique_width}
  For a graph $G$, its clique-width $\cliqueWidth{G}$ is the minimum number of labels needed to construct $G$ by the following four operations:
  \begin{itemize}
    \item $v(i):$ introduce a new vertex $v$ with label $i$.
    \item $\disUni:$ given two graphs $G$ and $H$, take their disjoint union $G \disUni H$.
    \item $\join{i}{j}\,(i \neq j):$ add edges between every two vertices with label $i$ and $j$.
    \item $\relabel{i}{j}:$ for all vertices with label $i$, change the labels into $j$.
  \end{itemize}
\end{definition}

The procedure to construct a graph $G$ by the above four operations can be associated with an algebraic expression.
Such an expression using at most $k$ labels is a \emph{$k$-expression}.
A tree associated with a $k$-expression is the \emph{$k$-expression tree}.
\cref{fig:tree} shows a $3$-expression tree for the graph in \cref{fig:unigraph}.

Cographs are exactly the graphs with clique-width at most 2~\cite{courcelle2000upper}.
In addition, a graph is a cograph if and only if it is $P_4$-free~\cite{brandstadt1999graph}.
Since a complement graph of a cograph is also a cograph~\cite{brandstadt1999graph}, the following lemma holds.

\begin{lemma}\label{le:cograph}
    If a graph $G$ is $P_4$-free, both $\cliqueWidth{G}$ and $\cliqueWidth{\compl{G}}$ are at most 2.
\end{lemma}

\begin{figure}[t]
    \centering
    \includegraphics[bb=0 0 527 480, scale=0.4]{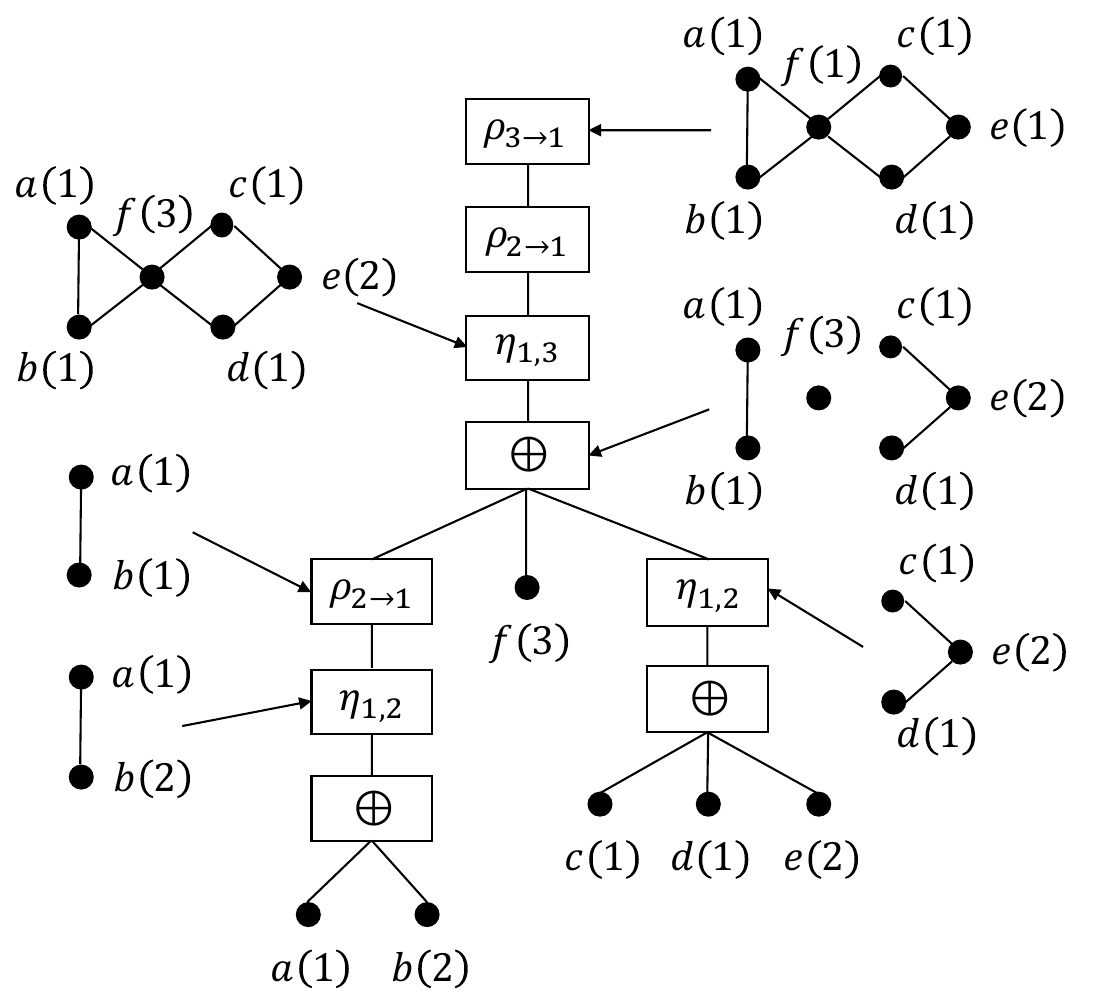}
    \caption{$3$-expression tree for the graph in \cref{fig:unigraph}. Each graph corresponds to the subtree rooted at the node pointed at by the arrow.
    $v(i)$ means that a vertex $v$ has the label $i$.
    The $3$-expression corresponding to the tree is $\relabel{3}{1}\relabel{2}{1}\join{1}{3}(\relabel{2}{1}\join{1}{2}(a(1) \disUni b(2)) \disUni f(3) \disUni \join{1}{2}(c(1) \disUni d(1) \disUni e(2)))$.}
    \label{fig:tree}
\end{figure}

\subsection{Canonical Decomposition and Characterization of Unigraphs}\label{sec:unigraph}
In this subsection, we introduce the canonical decomposition of Tyshkevich~\cite{tyshkevich2000decomposition} and a characterization of unigraphs.

\begin{definition}[Splitted graph]\label{def:splitted_graph}
  Let $G$ be a split graph with a bipartition $V(G) = A \cup B$, where $A$ is a clique and $B$ is an independent set.
  The triple $(G, A, B)$ is a \emph{splitted graph}.
\end{definition}

\begin{definition}[Composition]\label{def:composition}
  Let $(G, A, B)$ be a splitted graph and \linebreak
  $H = (V(H), E(H))$ be a simple graph.
  Their \emph{composition} is the graph $F = (G, A, B) \circ H$ such that:
    \begin{itemize}
        \item $V(F) = A \cup B \cup V(H)$, and
        \item $E(F) = E(G) \cup E(H) \cup \setST{\set{a, v}}{a \in A, v \in V(H)}$.
    \end{itemize}
\end{definition}

\begin{figure}[t]
    \centering
    \begin{subfigure}{.45\linewidth}
        \centering
        \includegraphics[bb=0 0 207 103, scale=0.5]{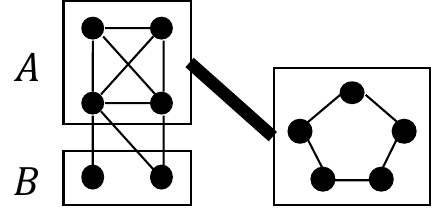}
        \caption{$G \circ H$.}
        \label{fig:composition1}
    \end{subfigure}
    \hfill
    \begin{subfigure}{.45\linewidth}
        \centering
        \includegraphics[bb=0 0 279 103, scale=0.5]{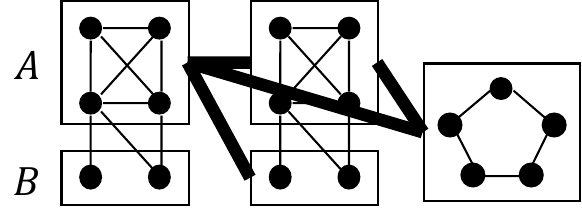}
        \caption{$G \circ G \circ H$.}
        \label{fig:composition2}
    \end{subfigure}
    \caption{Compositions of graphs. A bold line indicates that all the vertices in the endpoints are adjacent.}
    \label{fig:compositions}
\end{figure}

Intuitively, $(G, A, B) \circ H$ is a graph obtained by adding edges between every vertex in $A$ and $V(H)$ in $G \disUni H$.
\cref{fig:compositions} shows examples of compositions.
Note that a composition of two splitted graphs can be regarded as a splitted graph.
For two splitted graphs $(G_1, A_1, B_1)$ and $(G_2, A_2, B_2)$, their composition can be written as $(G, A_1 \cup A_2, B_1 \cup B_2)$.
Since the operation $\circ$ is associative, we omit parentheses when we use $\circ$ multiple times.
If a graph can be written as $(G, A, B) \circ H$, the graph is \emph{decomposable}; otherwise, \emph{indecomposable}.

\begin{theorem}[Decomposition theorem~\cite{tyshkevich2000decomposition}]
\label{th:decomposition_theorem}
    Every graph $G$ can be uniquely decomposed as $G = (G_k, A_k, B_k)\,\circ\, \dots\,\circ\,(G_1, A_1, B_1)\,\circ\,G_0$, where $(G_i, A_i, B_i)\,(i \in [k])$ is an indecomposable split graph and $G_0$ is an indecomposable nonsplit graph.
\end{theorem}

For a graph $G$, we call $(G_k, A_k, B_k) \circ \dots (G_1, A_1, B_1) \circ G_0$ the \emph{canonical decomposition} of $G$.
Tyshkevich gives a characterization of unigraphs based on \cref{th:decomposition_theorem}.
To explain it, we need some additional definitions.
For a splitted graph $(G, A, B)$, its \emph{complement} $\compl{(G, A, B)}$ is $(\compl{G}, B, A)$ and \emph{inverse} $\inv{(G, A, B)}$ is $(\inv{G}, B, A)$, where $\compl{G}$ is a complement graph of $G$ and $\inv{G}$ is the graph obtained by removing the edges in $\setST{\{a_1, a_2\}}{a_1, a_2 \in A}$ from $G$ and then adding the edges in $\setST{\{b_1, b_2\}}{b_1, b_2 \in B}$ to $G$.
In other words, $\inv{G}$ is the graph obtained from $G$ by inverting the existence of edges in the clique and the independent set.

We define the following graphs:
\begin{itemize}
    \item $U_2(m, s)$: It is the disjoint union of $mK_2\,(m \geq 1)$ and $K_{1, s}\,(s \geq 2)$. (\cref{fig:U2})
    \item $U_3(m)$: For $m \geq 1$, the graph is obtained by taking the disjoint union of $C_4$ and $m$ triangles $K_3$, choosing a vertex from each component, and merging all the vertices into one. (\cref{fig:U3})
    \item $S_2(p_1, q_1; \dots; p_t, q_t)$: For each $i \in [t]$, take the disjoint union of $q_i$ stars $K_{1, p_i}$ and add edges connecting every two centers of the stars, where $p_i, q_i, t \geq 1$ and $q_1 + \dots + q_t \geq 2$. (\cref{fig:S2})
    \item $S_3(p, q_1; q_2)$: For $S_2(p, q_1; p+1, q_2)$, where $p \geq 1, q_1 \geq 2$ and $q_2 \geq 1$, add a new vertex $v$ into the independent set and connect $v$ with the centers of $K_{1, p}$. (\cref{fig:S3})
    \item $S_4(p, q)$: For $S_3(p, 2; q)$, where $q \geq 1$, add a new vertex $u$ into the clique and connect $u$ with all the vertices other than $u$ and $v$. (\cref{fig:S4})
\end{itemize}

\begin{theorem}[Characterization of unigraphs~\cite{tyshkevich2000decomposition}]\label{th:unigraph}
    Unigraphs are the graphs can be written as $(G_k, A_k, B_k) \circ \dots \circ (G_1, A_1, B_1) \circ G_0$, where:
    \begin{itemize}
        \item $k \geq 0$ if $G_0 \neq \emptyset$ and $k \geq 1$ otherwise;
        \item For each $i \in [k]$, either $G_i$, $\compl{G_i}$, $\inv{G_i}$, or $\invCompl{G_i}$ is one of the following split unigraphs:
            \begin{equation}\label{eq:split_unigraphs}
                K_1,\quad S_2(p_1, q_1; \dots; p_t, q_t),\quad S_3(p, q_1; q_2),\quad S_4(p, q);
            \end{equation}
        \item If $G_0 \neq \emptyset$, either $G_0$ or $\compl{G_0}$ is one of the following nonsplit unigraphs:
            \begin{equation}\label{eq:non_split_unigraphs}
                C_5,\quad mK_2\ (m \geq 2),\quad U_2(m, s),\quad U_3(m).
            \end{equation}
    \end{itemize}
\end{theorem}

\begin{figure}[t]
    \centering
    \begin{subfigure}{.25\linewidth}
        \centering
        \includegraphics[bb=0 0 177 129, scale=0.4]{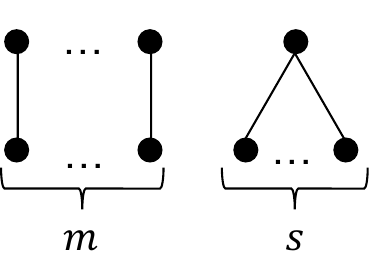}
        \caption{$U_2(m, s)$.}
        \label{fig:U2}
    \end{subfigure}
    \hfill
    \begin{subfigure}{.25\linewidth}
        \centering
        \includegraphics[bb=0 0 185 115, scale=0.4]{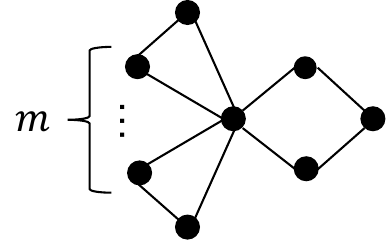}
        \caption{$U_3(m)$.}
        \label{fig:U3}
    \end{subfigure}
    \begin{subfigure}{.45\linewidth}
        \centering
        \includegraphics[bb=0 0 289 170, scale=0.4]{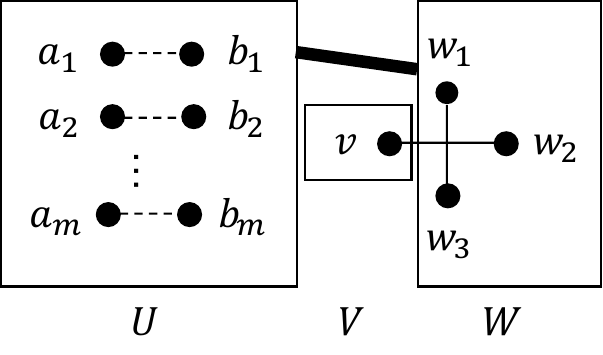}
        \caption{$\compl{U_3(m)}$. In $U$, two vertices are adjacent if and only if the vertices are not joined by a dotted line. The bold line indicates that all the vertices in $U$ and $W$ are adjacent each other.}
        \label{fig:U3compl}
    \end{subfigure}
    \caption{Indecomposable nonsplit unigraphs.}
    \label{fig:non_split_unigraph}
\end{figure}

\begin{figure}[t]
    \centering
    \begin{subfigure}{\linewidth}
        \centering
        \includegraphics[bb=0 0 369 145, scale=0.5]{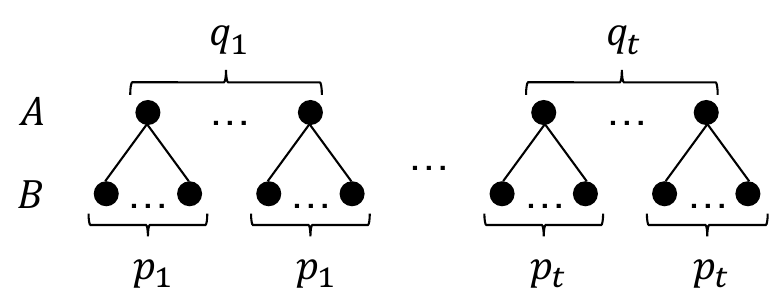}
        \caption{$S_2(p_1, q_1; \dots; p_t, q_t)$.}
        \label{fig:S2}
    \end{subfigure}
    \hfill
    \begin{subfigure}{\linewidth}
        \centering
        \includegraphics[bb=0 0 374 147, scale=0.5]{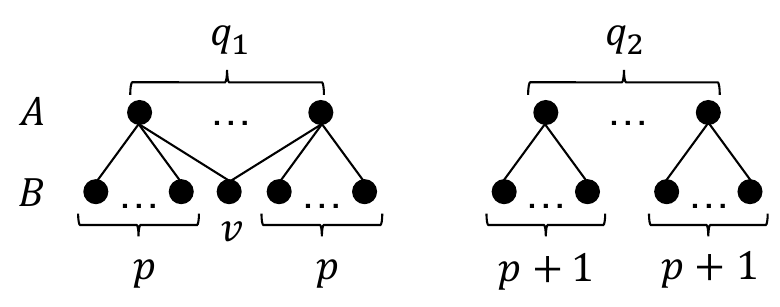}
        \caption{$S_3(p, q_1; q_2)$.}
        \label{fig:S3}
    \end{subfigure}
    \centering
    \begin{subfigure}{\linewidth}
        \centering
        \includegraphics[bb=0 0 374 147, scale=0.5]{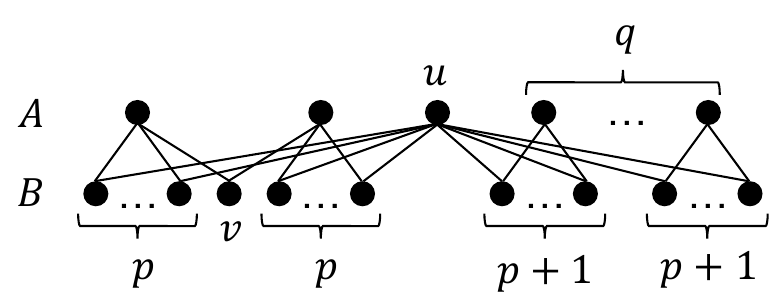}
        \caption{$S_4(p, q)$.}
        \label{fig:S4}
    \end{subfigure}
    \caption{Indecomposable split unigraphs. In each figure, the upper part is a clique and the lower part is an independent set. The edges connecting the vertices in the clique are omitted.}
    \label{fig:split_unigraph}
\end{figure}

\section{Clique-Width of Unigraphs}\label{sec:main}
In this section, we show that the clique-width of unigraphs is at most 4.
We prove our result by focusing on a relationship between \cref{th:unigraph} and the clique-width.

\begin{definition}[Split labeling]
    A splitted graph $(G, A, B)$ is \emph{split labeled} if all the vertices in the clique $A$ have the label $1$ and all the vertices in the independent set $B$ have the label $2$.
\end{definition}

\begin{definition}[Split clique-width]
    For a splitted graph $(G, A, B)$, its \emph{split clique-width} $\splitCliqueWidth{G}$ is the minimum number of labels needed to split label $G$ by the four operations in \cref{def:clique_width}.
    In addition, the \emph{$k$-split expression} is a $k$-expression to split label $G$ with at most $k$ labels and the \emph{$k$-split expression tree} is the corresponding expression tree.
\end{definition}

In the following, when the clique-width of $G$ is at most $k$, we use $\expr{G}{k}$ to denote an arbitrary $k$-expression of $G$. (We assume that the labels of all vertices are 1 after evaluating $\expr{G}{k}$.)
In addition, when the split clique-width of a splitted graph $(G, A, B)$ is at most $k$, we use $\splitExpr{G}{k}$ to denote an arbitrary $k$-split expression of $G$.

\begin{figure}[t]
    \centering
    \begin{subfigure}{.49\linewidth}
        \centering
        \includegraphics[bb=0 0 489 446, scale=0.3]{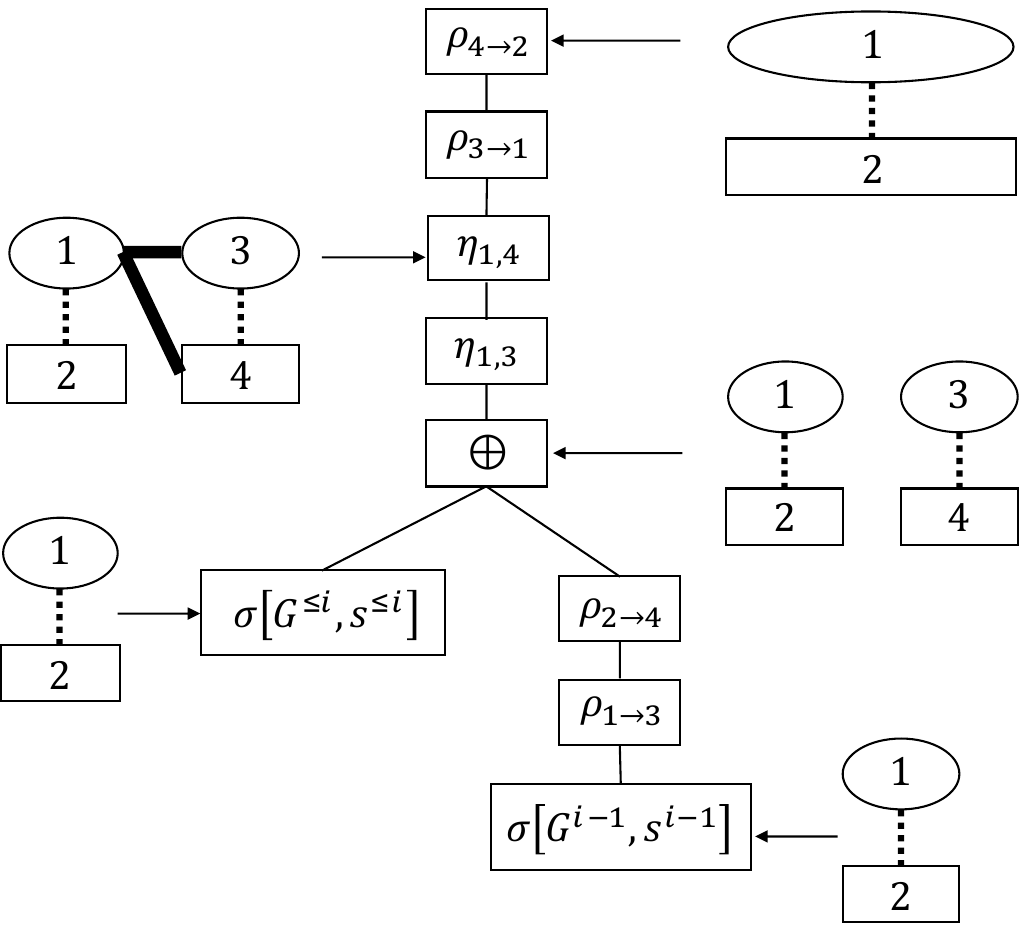}
        \caption{Split expression tree imitating the composition of splitted graphs.}
        \label{fig:composition_split_tree}
    \end{subfigure}
    \hfill
    \begin{subfigure}{.49\linewidth}
        \centering
        \includegraphics[bb=0 0 474 318, scale=0.3]{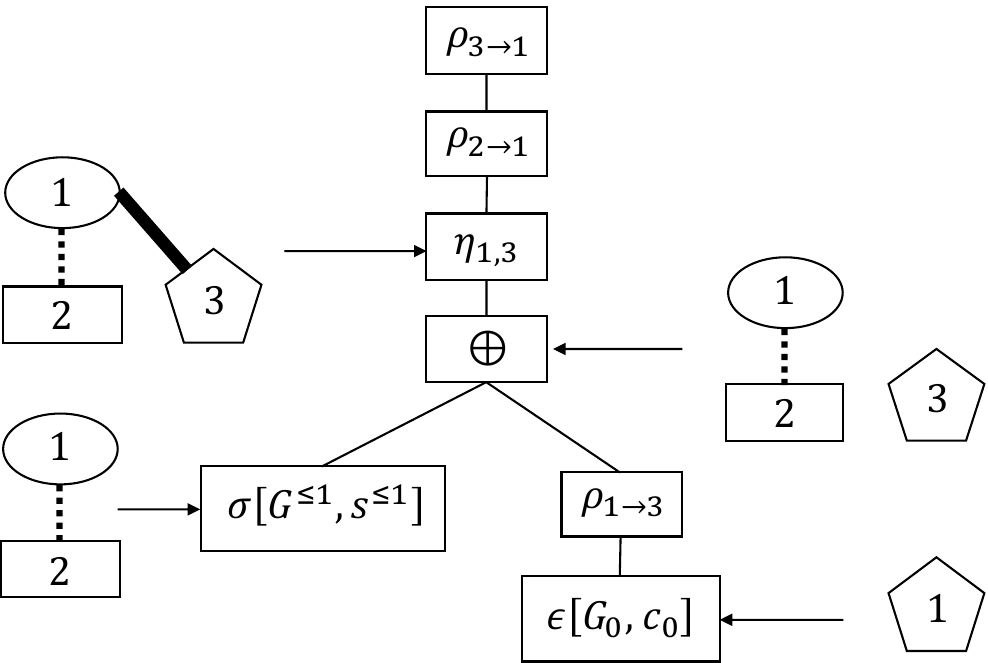}
        \caption{Expression tree imitating the composition of a splitted graph and a simple graph.}
        \label{fig:composition_nonsplit_tree}
    \end{subfigure}
    \caption{Expression trees imitating compositions of graphs. An ellipse is a clique and a square is an independent set. A pentagon is a simple graph. A bold line indicates that all the vertices in the endpoints are adjacent and a dotted line indicates that arbitrarily adjacent.}
    \label{fig:composition_trees}
\end{figure}

\begin{lemma}\label{le:clique_width_upper_bound}
    Let $G$ and $(G_k, A_k, B_k) \circ \dots \circ (G_1, A_1, B_1) \circ G_0$ be a unigraph and its canonical decomposition, respectively. Then,
    \begin{equation}\label{eq:upper_bound}
        \cliqueWidth{G} \leq \max \set{4, \max_{i \in [k]} \splitCliqueWidth{G_i}, \cliqueWidth{G_0}}.
    \end{equation}
\end{lemma}
\begin{proof}
    Let $c_0 = \cliqueWidth{G_0}$, $s_i = \splitCliqueWidth{G_i}$, and $G^{\leq i} = G_k \circ \dots \circ G_i$ for each $i \in [k]$.
    Since a composition of splitted graphs can be regarded as a splitted graph, $G_i$ is a splitted graph for each $i \in [k]$.
    Thus, we define $s^{\leq i} = \splitCliqueWidth{G^{\leq i}}$ for each $i \in [k]$.
    If $k = 0$ ($G = G_0$), then $\cliqueWidth{G} = c_0$, and thus \eqref{eq:upper_bound} holds.
    We assume $k > 0$ in the following.

    For each $i \in [k]$, we show the following by induction:
    \begin{equation}\label{eq:split_upper_bound}
        s^{\leq i} \leq \max \set{4, \max_{i \leq j \leq k} s_j}.
    \end{equation}
    First, by definition, $s^{\leq k} \leq s_k$ holds.
    Next, for an integer $i$ with $2 \leq i \leq k$, assume that \eqref{eq:split_upper_bound} holds.
    Then, we can construct a split expression of $G^{\leq i-1} = G^{\leq i} \circ G_{i-1}$ using $\splitExpr{G^{\leq i}}{s^{\leq i}}$ and $\splitExpr{G_{i-1}}{s_{i-1}}$:
    \begin{equation}\label{eq:expr_composition_split}
        \relabel{4}{2}\,\relabel{3}{1}\,\join{1}{4}\,\join{1}{3}
            \left( \splitExpr{G^{\leq i}}{s^{\leq i}} \disUni \relabel{2}{4}\,\relabel{1}{3}\,\splitExpr{G_{i-1}}{s_{i-1}} \right).
    \end{equation}
    \cref{fig:composition_split_tree} shows the corresponding split expression tree.
    The number of labels used in \eqref{eq:expr_composition_split} is (including the labels inside $\splitExpr{G^{\leq i}}{s^{\leq i}}$ and $\splitExpr{G_{i-1}}{s_{i-1}}$):
    \begin{align}
        \max \set{4, s^{\leq i}, s_{i-1}}
        &\leq \max \set{4, \set{4, \max_{i \leq j \leq k} s_j}, s_{i-1}} \nonumber \\
        &=    \max \set{4, \max_{i - 1 \leq j \leq k} s_j}.
    \end{align}
    Therefore, \eqref{eq:split_upper_bound} holds for $i-1$.
    By induction, \eqref{eq:split_upper_bound} holds also for $i = 1$.
    If $G_0 = \emptyset$, we have proven \eqref{eq:upper_bound}.
    If $G_0 \neq \emptyset$, we can construct a split expression of $G = G^{\leq 1} \circ G_0$ using $\splitExpr{G^{\leq 1}}{s^{\leq 1}}$ and $\expr{G_0}{c_0}$:
    \begin{equation}\label{eq:expr_composition_nonsplit}
        \relabel{3}{1}\relabel{2}{1}\,\join{1}{3}
          \left( \splitExpr{G^{\leq 1}}{s^{\leq 1}} \disUni
             \relabel{1}{3}\,\expr{G_0}{c_0}
             \right).
    \end{equation}
    \cref{fig:composition_nonsplit_tree} shows the corresponding expression tree.
     The number of labels used in \eqref{eq:expr_composition_nonsplit} is (including the labels inside $\splitExpr{G^{\leq 1}}{s^{\leq 1}}$ and $\expr{G_0}{c_0}$):
    \begin{align}
        \max \set{3, s^{\leq 1}, c_0}
        &\leq \max \set{3, \set{4, \max_{i \in [k]} s_i}, c_0} \nonumber \\
        &=    \max \set{4, \max_{i \in [k]} s_i, c_0}.
    \end{align}
    Therefore, \eqref{eq:upper_bound} holds.
\qed\end{proof}

By \cref{le:clique_width_upper_bound}, to prove $\cliqueWidth{G} \le 4$, it suffices to show that $\cliqueWidth{G_0}$ and $\splitCliqueWidth{G_i}\ (i \in [k])$ are at most 4.

\begin{figure}[t]
    \centering
    \begin{subfigure}{.99\linewidth}
        \centering
        \includegraphics[bb=0 0 575 341, scale=0.3]{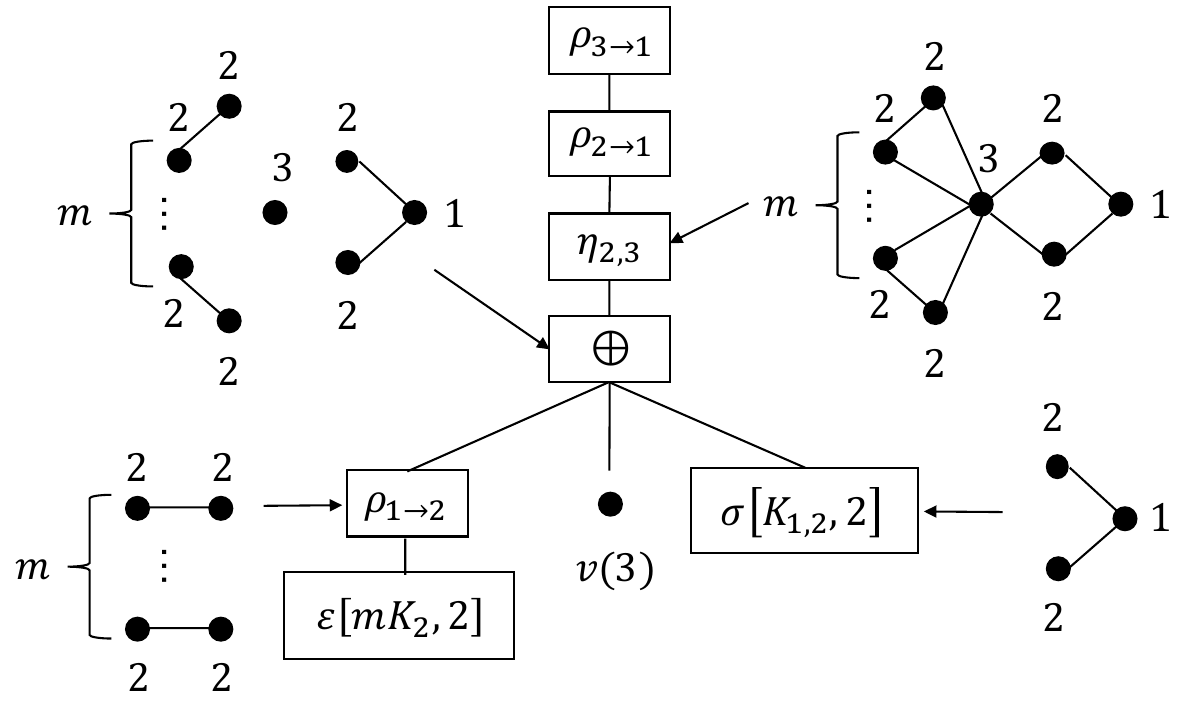}
        \caption{$3$-expression tree for $U_3(m)$.}
        \label{fig:U3_tree}
    \end{subfigure}
    \hfill
    \begin{subfigure}{.99\linewidth}
        \centering
        \includegraphics[bb=0 0 635 447, scale=0.3]{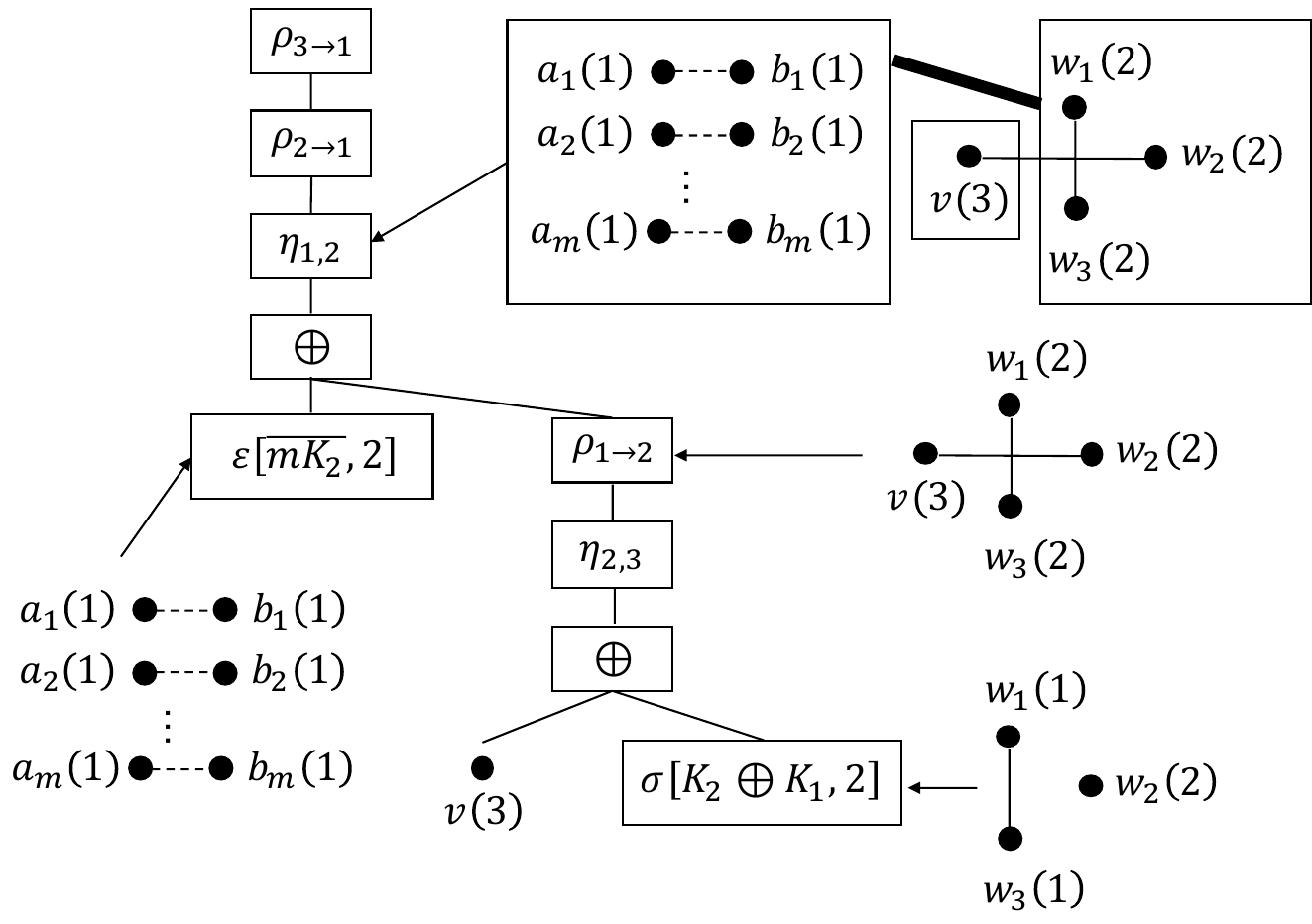}
        \caption{$3$-expression tree for $\compl{U_3(m)}$.
        In $\compl{mK_2}$, two vertices are adjacent if and only if they are not joined by a dotted line.}
        \label{fig:U3compl_tree}
    \end{subfigure}
    \caption{Expression trees for indecomposable nonsplit unigraphs.}
    \label{fig:non_split_tree}
\end{figure}

\begin{lemma}\label{le:non_split_unigraph}
    $\cliqueWidth{G_0} \leq 3$.
\end{lemma}
\begin{proof}
    By \cref{th:unigraph}, either $G_0$ or $\compl{G_0}$ is one of the graphs in \eqref{eq:non_split_unigraphs}.
    We show the clique-widths of the graphs in \eqref{eq:non_split_unigraphs} and their complement graphs are at most 4.
    \cref{tab:clique_width_upper_bound} in \cref{app:upper} summarizes the results shown in the below.

    The clique-width of $C_5\ (= \compl{C_5})$ is exactly 3~\cite{courcelle2000upper}.
    Since $mK_2$ and $U_2(m, s)$ are $P_4$-free, by \cref{le:cograph}, all of $\cliqueWidth{mK_2}$, $\cliqueWidth{U_2(m, s)}$, $\cliqueWidth{\compl{mK_2}}$, and $\cliqueWidth{\compl{U_2(m, s)}}$ are at most 2.

    Next, we consider $U_3(m)$.
    We can construct a $3$-expression of $U_3(m)$ using $\expr{mK_2}{2}$ and $\splitExpr{K_{1,2}}{2}$:
    \begin{equation}\label{eq:U3_expr}
        \relabel{3}{1}\,\relabel{2}{1}\,\join{2}{3} \left(
          \relabel{1}{2}\,\expr{mK_2}{2} \disUni v(3) \disUni \splitExpr{K_{1,2}}{2} \right).
    \end{equation}
    Therefore, $\cliqueWidth{U_3(m)} \leq 3$.
    \cref{fig:U3_tree} shows a $3$-expression tree corresponding to the $3$-expression in \eqref{eq:U3_expr}.

    Finally, we consider $\compl{U_3(m)}$, which is shown in \cref{fig:U3compl}.
    The vertex set can be partitioned into three sets $U = \set{a_1, b_1, \dots, a_m, b_m}$, $V = \set{v}$, and $W = \set{w_1, w_2, w_3}$.
    The edge set is $\setST{\edge{a_i}{b_j}}{i \neq j} \cup \setST{\edge{u}{w}}{u \in U, w \in W} \cup \set{\edge{v}{w_2}, \edge{w_1}{w_3}}$.
    Observe that $\compl{U_3(m)}[U]$ is isomorphic to $\compl{mK_2}$ and $\compl{U_3(m)}[W]$ is isomorphic to $K_2 \disUni K_1$.
    Therefore, the following $3$-expression constructs $\compl{U_3(m)}$:
    \begin{equation}\label{eq:U3_compl_expr}
        \relabel{3}{1}\,\relabel{2}{1}\,\join{1}{2}
        \left(
        \begin{array}{c}
            \expr{\compl{mK_2}}{2} \\
              \disUni \\
            \relabel{1}{2}\,\join{2}{3}\,\left( v(3) \disUni \splitExpr{K_2 \disUni K_1}{2} \right)
        \end{array}
        \right).
    \end{equation}
    It follows that $\cliqueWidth{\compl{U_3(m)}} \leq 3$.
    \cref{fig:U3compl_tree} shows the corresponding $3$-expression tree to the $3$-expression in \eqref{eq:U3_compl_expr}.\footnote{Since neither $C_5\ (= \compl{C_5})$, $U_3(m)$ nor $\compl{U_3(m)}$ are $P_4$-free, the clique-widths of these graphs are at least 3. Therefore, the upper bounds for these graphs are tight. 
    }
\qed\end{proof}

\begin{lemma}\label{le:split_unigraph}
    For each $i \in [k]$, $\splitCliqueWidth{G_i} \leq 4$.
\end{lemma}
\begin{proof}
    By \cref{th:unigraph}, for each $i \in [k]$, either $G_i, \compl{G_i}, \inv{G_i}$, or $\compl{\inv{G_i}}$ is one of the graphs in \eqref{eq:split_unigraphs}.
    We show that the clique-width of each graph is at most 4.
    Note that, when either $G_i$, $\compl{G_i}$, $\inv{G_i}$, or $\compl{\inv{G_i}}$ is $K_1$, it is easy to split label it by one label.
    When the splitted graph is $(K_1, \set{a}, \emptyset)$ (resp.\ $(K_1, \emptyset, \set{a})$), we introduce $a$ with label 1 (resp.\ 2).
    In the following we consider $S_2$, $S_3$, and $S_4$.
    First, we show that the clique-width of $S_2$, $\compl{S_2}$, $\inv{S_2}$, and $\invCompl{S_2}$ are at most 3 or 4 by induction.
    Next, for $S_3$ and $S_4$, we construct $4$-split expressions from $S_2$.
    Similarly, we construct $4$-split expressions for $\compl{S_3}$ and $\compl{S_4}$ from $\compl{S_2}$, and the same goes to $\inv{S_3}$, $\inv{S_4}$, $\invCompl{S_3}$, and $\invCompl{S_4}$.
    \cref{tab:split_clique_width_upper_bound} in \cref{app:upper} summarizes the results shown in the below.

    The graph $S_2$ is obtained by taking disjoint union of $l = \sum_{i=1}^{t} q_i$ stars and adding edges between every two centers of the stars.
    Let $p'_1, \dots, p'_l$ be the non-decreasing sequence of degrees of the centers of the stars.
    We write $S_2(p_1, q_1; \dots; p_t, q_t)$ as $S_2(p'_1, \dots, p'_l)$.
    We define $S_2^{\leq i} = S_2(p'_1, \dots, p'_i)$.
    For each $i \in [l]$, $S_2^{\leq i}$ is a split graph.
    We show that, for each $i \in [l]$, the split clique-width of $S_2^{\leq i}$ is at most 3 by induction.
    First, $S_2^{\leq 1} = K_{1, p'_1}$ is a star.
    In general, $K_{1, n}$ is constructed by the following $2$-split expression:
    \begin{equation}
        \join{1}{2} \left( u(1) \disUni
          v_1(2) \disUni \dots \disUni v_n(2) \right),
    \end{equation}
    where $u$ is the center of the star and $v_i\ (i \in [n])$ is a leaf.
    Next, for an integer $i \in [l - 1]$, assume that $\splitCliqueWidth{S_2^{\leq i}} \leq 3$.
    Then, $S_2^{\leq i+1}$ can be constructed by the following $3$-split expression using $\splitExpr{S_2^{\leq i}}{3}$ and $\splitExpr{K_{1, p'_{i+1}}}{2}$:
    \begin{equation}\label{eq:expr_S2}
        \relabel{3}{1}\,\join{1}{3}
          \left( \splitExpr{S_2^{\leq i}}{3} \disUni
             \relabel{1}{3}\,\splitExpr{K_{1, p'_{i+1}}}{2} \right).
    \end{equation}
    Therefore, $\splitCliqueWidth{S_2^{\leq i+1}} \leq 3$.
    \cref{fig:S2_tree} in \cref{app:exp_tree} shows the corresponding expression tree.
    By induction, $\splitCliqueWidth{S_2^{\leq l}} \leq 3$ holds.

    The graph $\inv{S_2}$ is obtained by taking disjoint union of $l$ stars and adding edges between the leaves of the stars.
    We define $\inv{S_2}^{\leq i}$ in the same way as $S_2^{\leq i}$.
    We show that $\splitCliqueWidth{\inv{S_2}^{\leq i}} \leq 3$ for each $i \in [l]$ by induction.
    First, $\inv{S_2}^{\leq 1} = \inv{K_{1, p'_1}}$ is a graph obtained by adding a vew vertex $w$ to $K_{p'_1}$ and edges between $w$ and the other vertices, that is, $K_{p'_1 + 1}$.
    In general, $K_{n+1}$ is constructed by the following $2$-split expression:
    \begin{equation}
        \join{1}{2} \left( w(2) \disUni \expr{K_{n}}{2} \right).
    \end{equation}
    Next, for an integer $i \in [l - 1]$, assume that $\splitCliqueWidth{\inv{S_2}^{\leq i}} \leq 3$.
    Then, $\inv{S_2}^{\leq i+1}$ can be constructed by the following $3$-split expression using $\splitExpr{\inv{S_2}^{\leq i}}{3}$ and $\splitExpr{\inv{K_{1, p'_{i+1}}}}{2}$:
    \begin{equation}\label{eq:expr_S2inv}
        \relabel{3}{1}\,\join{1}{3}
          \left( \splitExpr{\inv{S_2}^{\leq i}}{3} \disUni
            \relabel{1}{3}\,\splitExpr{\inv{K_{1, p'_{i+1}}}}{2} \right)
    \end{equation}
    \cref{fig:S2inv_tree} in \cref{app:exp_tree} shows the corresponding expression tree. By induction, $\splitCliqueWidth{\inv{S_2}^{\leq l}} \leq 3$ holds.

    The graph $\compl{S_2}$ is the following graph.
    The vertex set can be partitioned into $C = \set{c_1, \dots, c_l}, L_1, \dots, L_l$, where $L_i$ contains $p_i'$ vertices.
    Two vertices $u$ and $v$ are adjacent if and only if (a) for some integers $i$ and $j$ with $i \neq j$, both $u = c_i$ and $v \in L_j$ hold, or (b) for some integers $i$ and $j$ (possibly $i = j$), $u \in L_i$ and $v \in L_j$ hold.
    For each integer $i \in [l]$, we define $\compl{S_2}^{\leq i}$ as an induced subgraph of $\compl{S_2}$ by the set $\set{c_1, \dots, c_i} \cup L_1 \cup \dots \cup L_i$ of vertices.
    We show that $\splitCliqueWidth{\compl{S_2}^{\leq i}} \leq 4$ for each $i \in [l]$ by induction.
    First, $\compl{S_2}^{\leq 1}$ is the disjoint union of an isolated vertex $c_1$ and the complete graph induced by $L_1$, that is, $\compl{S_2}^{\leq 1} = \compl{K_{1, p'_1}} = K_1 \disUni K_{p'_1}$.
    In general, $\compl{K_{1, n}}$ is constructed by the following $2$-split expression:
    \begin{equation}
        u(2) \disUni \expr{K_{n}}{2},
    \end{equation}
    where $u$ is the isolated vertex (the only vertex of $K_1$).
    Next, for an integer $i \in [l - 1]$, assume that $\splitCliqueWidth{\compl{S_2}^{\leq i}} \leq 4$.
    Then, $\compl{S_2}^{\leq i}$ is constructed by the following $4$-split expression using $\splitExpr{\compl{S_2}^{\leq i}}{4}$ and $\splitExpr{\compl{K_{1, p'_{i+1}}}}{2}$:
    \begin{equation}\label{eq:expr_S2compl}
        \relabel{4}{2}\,\relabel{3}{1}\,\join{2}{3}\,\join{1}{4}\,\join{1}{3}
        \left( \splitExpr{\compl{S_2}^{\leq i}}{4} \disUni
        \relabel{2}{4}\,\relabel{1}{3}\,\splitExpr{\compl{K_{1, p'_{i+1}}}}{2} \right).
    \end{equation}
    \cref{fig:S2compl_tree} in \cref{app:exp_tree} shows the corresponding expression tree. By induction, $\splitCliqueWidth{\compl{S_2}^{\leq l}} \leq 4$ holds.

    The graph $\invCompl{S_2}$ is the following graph.
    Similarly to $\compl{S_2}$, the vertex set can be partitioned into $C = \set{c_1, \dots, c_l}, L_1, \dots, L_l$.
    Two vertices $u$ and $v$ are adjacent if and only if (a) for some integers $i$ and $j$ such that $i \neq j$, both $u = c_i$ and $v \in L_j$ hold, or (b) for some integers $i$ and $j$ (possibly $i = j$), $u, v \in C$ holds.
    We define $\invCompl{S_2}^{\leq i}$ in the same way as $\compl{S_2}^{\leq i}$.
    We show that $\splitCliqueWidth{\invCompl{S_2}^{\leq i}} \leq 4$ for each $i \in [l]$ by induction.
    $\invCompl{S_2}^{\leq 1}$ is the disjoint union of the isolated vertex $c_1$ and $|L_1| = p_1'$ isolated vertices, that is, $\invCompl{S_2}^{\leq 1} = \invCompl{K_{1, p_1'}} = K_1 \disUni \compl{K_{p_1'}}$.
    In general, $\invCompl{K_{1, n}}$ is constructed by the following $2$-split expression:
    \begin{equation}
        u(1) \disUni v_1(2) \disUni \dots \disUni v_n(2).
    \end{equation}
    Next, for an integer $i \in [l - 1]$, assume that $\splitCliqueWidth{\invCompl{S_2}^{\leq i}} \leq 4$.
    Then, $\invCompl{S_2}^{\leq i+1}$ is constructed by the following $4$-split expression using $\splitExpr{\invCompl{S_2}^{\leq i}}{4}$ and $\splitExpr{\invCompl{K_{1, p'_{i+1}}}}{2}$:
    \begin{equation}\label{eq:expr_S2invCompl}
        \relabel{4}{2}\,\relabel{3}{1}\,\join{2}{3}\,\join{1}{4}\,\join{1}{3}
        \left( \splitExpr{\invCompl{S_2}^{\leq i}}{4} \disUni
        \relabel{2}{4}\,\relabel{1}{3}\,\splitExpr{\invCompl{K_{1, p'_{i+1}}}}{2} \right).
    \end{equation}
    \cref{fig:S2invCompl_tree} in \cref{app:exp_tree} shows the corresponding expression tree. By induction, $\splitCliqueWidth{\invCompl{S_2}^{\leq l}} \leq 4$.

    $S_3$, $\inv{S_3}$, $\compl{S_3}$, and $\invCompl{S_3}$ are constructed by the following 3, 3, 4, and 4-split expressions, respectively:
    \begin{equation}
        \relabel{3}{1}\,\join{1}{3} \left(
          \begin{array}{c}
            \relabel{3}{2}\,\join{1}{3}
              \left( \splitExpr{S_2(p, q_1)}{3} \disUni v(3) \right) \\
            \disUni \\
            \relabel{1}{3}\,\splitExpr{S_2(p+1, q_2)}{3}
          \end{array}
        \right),
    \end{equation}
    \begin{equation}
        \relabel{3}{1}\,\join{1}{3} \left(
          \begin{array}{c}
            \join{2}{3}
              \left( \splitExpr{\inv{S_2}(p, q_1)}{3} \disUni v(3) \right) \\
            \disUni \\
            \relabel{1}{3}\,\splitExpr{\inv{S_2}(p+1, q_2)}{3}
          \end{array}
        \right),
    \end{equation}
    \begin{equation}
        \relabel{4}{2}\,\relabel{3}{1}\,\join{2}{3}\,\join{1}{4}\,\join{1}{3} \left(
          \begin{array}{c}
            \relabel{3}{1}\,\join{1}{3}
              \left( \splitExpr{\compl{S_2}(p, q_1)}{4} \disUni v(3) \right) \\
            \disUni \\
            \relabel{2}{4}\,\relabel{1}{3}\,\splitExpr{\compl{S_2}(p+1, q_2)}{4}
          \end{array}
        \right),
    \end{equation}
    \begin{equation}
        \relabel{4}{2}\,\relabel{3}{1}\,\join{2}{3}\,\join{1}{4}\,\join{1}{3} \left(
          \begin{array}{c}
            \left( \splitExpr{\invCompl{S_2}(p, q_1)}{4} \disUni v(2) \right) \\
            \disUni \\
            \relabel{2}{4}\,\relabel{1}{3}\,\splitExpr{\invCompl{S_2}(p+1, q_2)}{4}
          \end{array}
        \right).
    \end{equation}
    Therefore, $\splitCliqueWidth{S_3}$, $\splitCliqueWidth{\inv{S_3}}$, $\splitCliqueWidth{\compl{S_3}}$, and $\splitCliqueWidth{\invCompl{S_3}}$ are at most $3, 3, 4$, and $4$, respectively.

    $S_4$, $\inv{S_4}$, $\compl{S_4}$, and $\invCompl{S_4}$ are constructed by the following 4-split expressions:
    \begin{align}
        \relabel{4}{2}\,\relabel{3}{1}\,\join{2}{3}\,\join{1}{3} \left(
            u(3) \disUni X_1 \right), \\
        \relabel{4}{2}\,\relabel{3}{1}\,\join{1}{4} \left(
            u(4) \disUni X_2 \right), \\
        \relabel{4}{2}\,\relabel{3}{1}\,\join{3}{4}\,\join{1}{3}\left(
            u(4) \disUni X_3 \right), \\
        \relabel{4}{2}\,\relabel{3}{1}\,\join{3}{4}\,\join{1}{3} \left(
            u(3) \disUni X_4 \right),
    \end{align}
    where
    \begin{align}
        X_1 &=
          \relabel{3}{1}\,\join{1}{4}\,\left(
          v(4)\,\disUni\,\relabel{4}{2}\,\join{1}{3} \left(
          \begin{array}{c}
            \splitExpr{S_2(p, 2)}{4} \\
            \disUni \\
            \relabel{2}{4}\,\relabel{1}{3}\,\splitExpr{S_2(p+1, q)}{4}
          \end{array}
        \right) \right), \\
        X_2 &=
          \relabel{4}{2}\,\join{2}{3}\,\left(
          v(3)\,\disUni\,\relabel{3}{1}\,\join{1}{3} \left(
          \begin{array}{c}
            \splitExpr{\inv{S_2}(p, 2)}{4} \\
            \disUni \\
            \relabel{2}{4}\,\relabel{1}{3}\,\splitExpr{\inv{S_2}(p+1, q)}{4}
          \end{array}
        \right) \right), \\
        X_3 &=
          \relabel{4}{2}\,\join{3}{4}\,\left(
          v(3)\,\disUni\,\relabel{3}{1}\,\join{1}{3} \left(
          \begin{array}{c}
            \splitExpr{\compl{S_2}(p, 2)}{4} \\
            \disUni \\
            \relabel{2}{4}\,\relabel{1}{3}\,\splitExpr{\compl{S_2}(p+1, q)}{4}
          \end{array}
        \right) \right), \\
        X_4 &=
          \relabel{3}{1}\,\join{3}{4}\,\left(
          v(4)\,\disUni\,\relabel{4}{2}\,\join{1}{3} \left(
          \begin{array}{c}
            \splitExpr{\invCompl{S_2}(p, 2)}{4} \\
            \disUni \\
            \relabel{2}{4}\,\relabel{1}{3}\,\splitExpr{\invCompl{S_2}(p+1, q)}{4}
          \end{array}
        \right) \right).
    \end{align}
    Therefore, $\splitCliqueWidth{S_4}$, $\splitCliqueWidth{\inv{S_4}}$, $\splitCliqueWidth{\compl{S_4}}$, and $\splitCliqueWidth{\invCompl{S_4}}$ are at most $4$.
\qed\end{proof}

Now we state our main theorem.

\begin{theorem}\label{th:main_theorem}
    If $G$ is a unigraph, its clique-width is at most 4.
\end{theorem}
\begin{proof}
    By Lemmas~\ref{le:clique_width_upper_bound}--\ref{le:split_unigraph},
    $\cliqueWidth{G} \leq \max\{4, 4, 3\} = 4$ holds.
\qed\end{proof}

The upper bound 4 is tight: the graph $\compl{S_2(4,4)}$ has clique-width 4, as can be checked using a software~\cite{heule2015sat}.

\subsubsection*{Acknowledgements.}
We thank Konrad Dabrowski for valuable comments to our preprint.
We thank Jun Kawahara and Shin-ichi Minato for fruitful discussion.
This work is partially supported by JSPS KAKENHI Grant Number JP19H01103 and JP19J21000.

%
%
%
\bibliographystyle{splncs04}
\bibliography{reference}

\appendix

\section{Upper bounds on clique-width and split clique-width}\label{app:upper}
\cref{tab:clique_width_upper_bound} summarizes the upper bounds on clique-width of indecomposable nonsplit unigraphs.
\cref{tab:split_clique_width_upper_bound} summarizes the upper bounds on split clique-width of indecomposable split unigraphs.

\begin{table}[t]
\centering
\caption{Upper bounds on the clique-width of indecomposable nonsplit unigraphs.}
\begin{tabular}{c|c}
graph                & upper bound on the clique-width \\ \hline

$C_5 (= \compl{C_5})$ & 3                     \\

$mK_2$               & 2                      \\
$\compl{mK_2}$       & 2                      \\

$U_2(m, s)$          & 2                      \\
$\compl{U_2(m, s)}$  & 2                      \\

$U_3(m)$             & 3                      \\
$\compl{U_3(m)}$     & 3                      \\
\end{tabular}
\label{tab:clique_width_upper_bound}
\end{table}

\begin{table}[t]
\centering
\caption{Upper bounds on the split clique-width of indecomposable split unigraphs.}
\begin{tabular}{c|c}
graph                & upper bound on the split clique-width\\ \hline

$S_2$                & 3                      \\
$\inv{S_2}$          & 3                      \\
$\compl{S_2}$        & 4                      \\
$\invCompl{S_2}$     & 4                      \\ \hline

$S_3$                & 3                      \\
$\inv{S_3}$          & 3                      \\
$\compl{S_3}$        & 4                      \\
$\invCompl{S_3}$     & 4                      \\ \hline

$S_4$                & 4                      \\
$\inv{S_4}$          & 4                      \\
$\compl{S_4}$        & 4                      \\
$\invCompl{S_4}$     & 4                      \\
\end{tabular}
\label{tab:split_clique_width_upper_bound}
\end{table}

\section{Expression trees}\label{app:exp_tree}
\cref{fig:S2_tree,fig:S2inv_tree,fig:S2compl_tree,fig:S2invCompl_tree} show expression trees for $S_2$, $\inv{S_2}$, $\compl{S_2}$, and $\invCompl{S_2}$, which correspond respectively to the expressions in \eqref{eq:expr_S2}, \eqref{eq:expr_S2inv}, \eqref{eq:expr_S2compl}, and \eqref{eq:expr_S2invCompl}.

\begin{figure}[t]
    \centering
    \includegraphics[bb=0 0 660 334, scale=0.4]{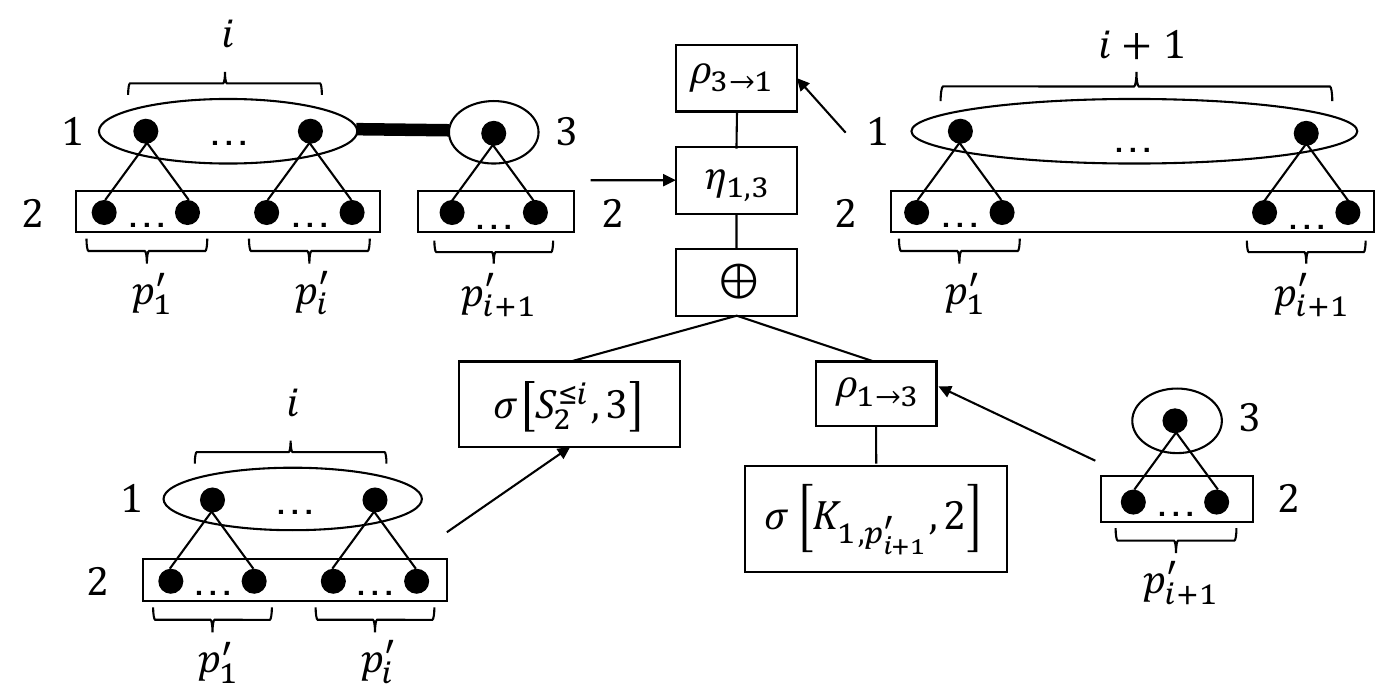}
    \caption{Expression tree for $S_2$. An ellipse is a clique and a rectangle is an independent set. An integer beside an ellipse or a rectangle is the label of the vertices in it. A bold line indicates that all the vertices in the endpoints are adjacent.}
    \label{fig:S2_tree}
\end{figure}

\begin{figure}[t]
    \centering
    \includegraphics[bb=0 0 692 335, scale=0.4]{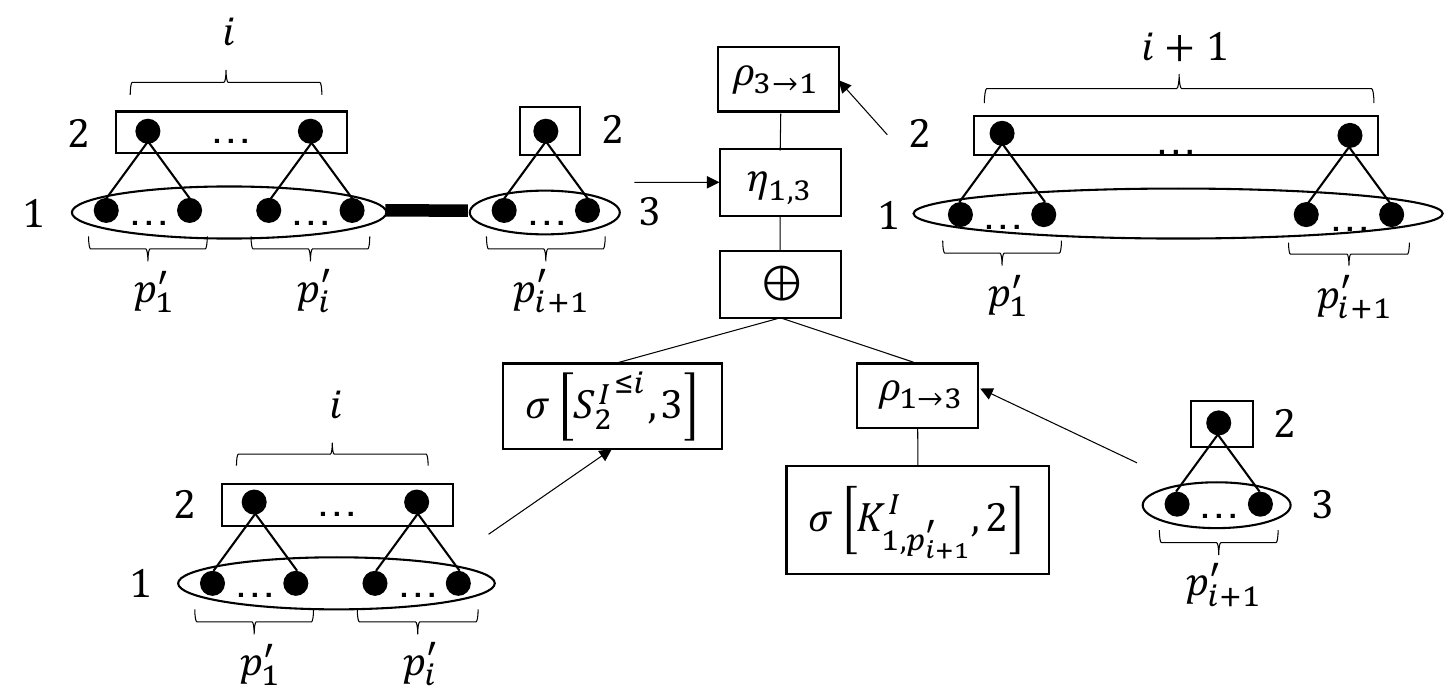}
    \caption{Expression tree for $\inv{S_2}$.}
    \label{fig:S2inv_tree}
\end{figure}

\begin{figure}[t]
    \centering
    \includegraphics[bb=0 0 882 478, scale=0.4]{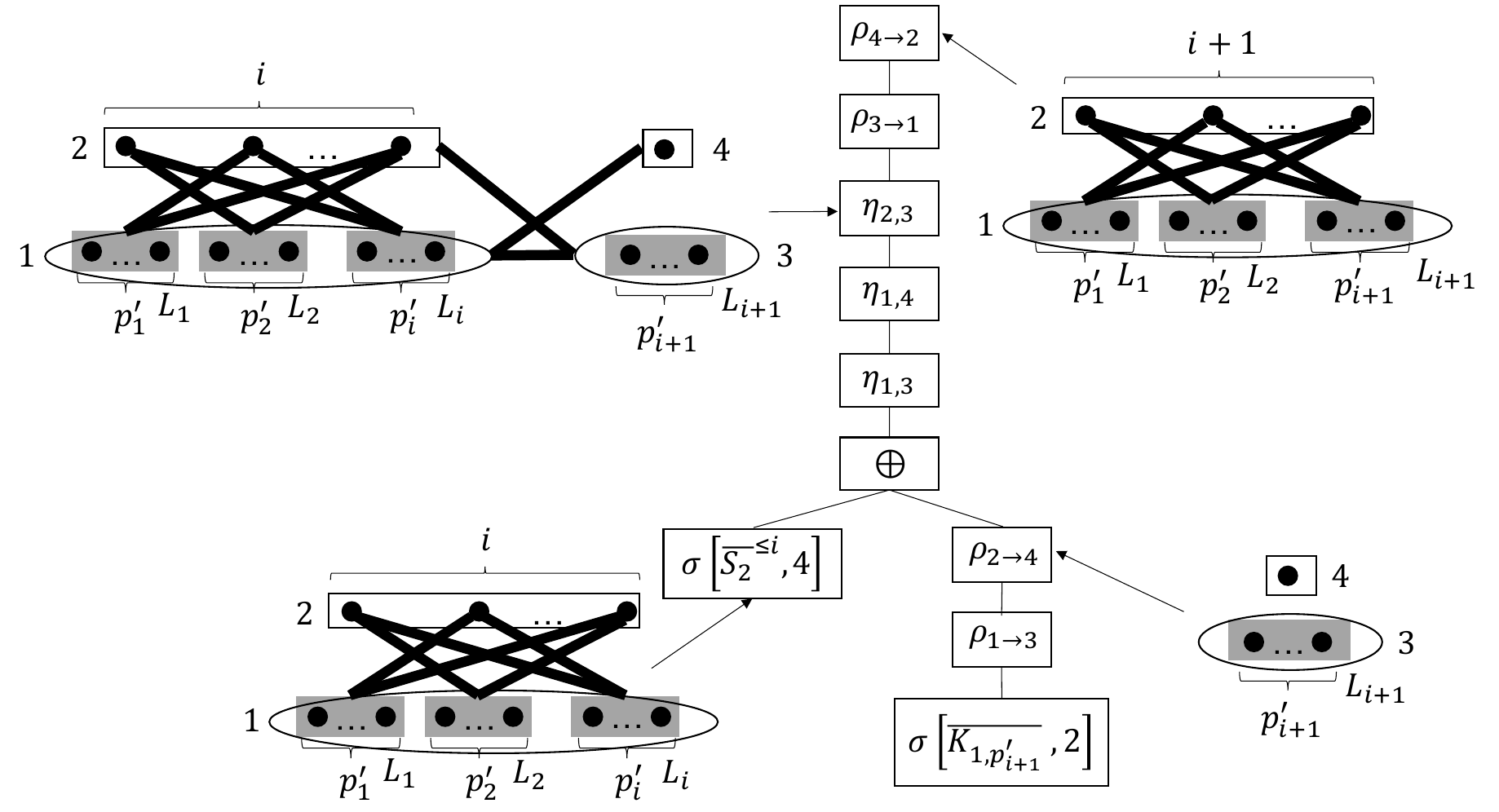}
    \caption{Expression tree for $\compl{S_2}$.}
    \label{fig:S2compl_tree}
\end{figure}

\begin{figure}[t]
    \centering
    \includegraphics[bb=0 0 816 480, scale=0.4]{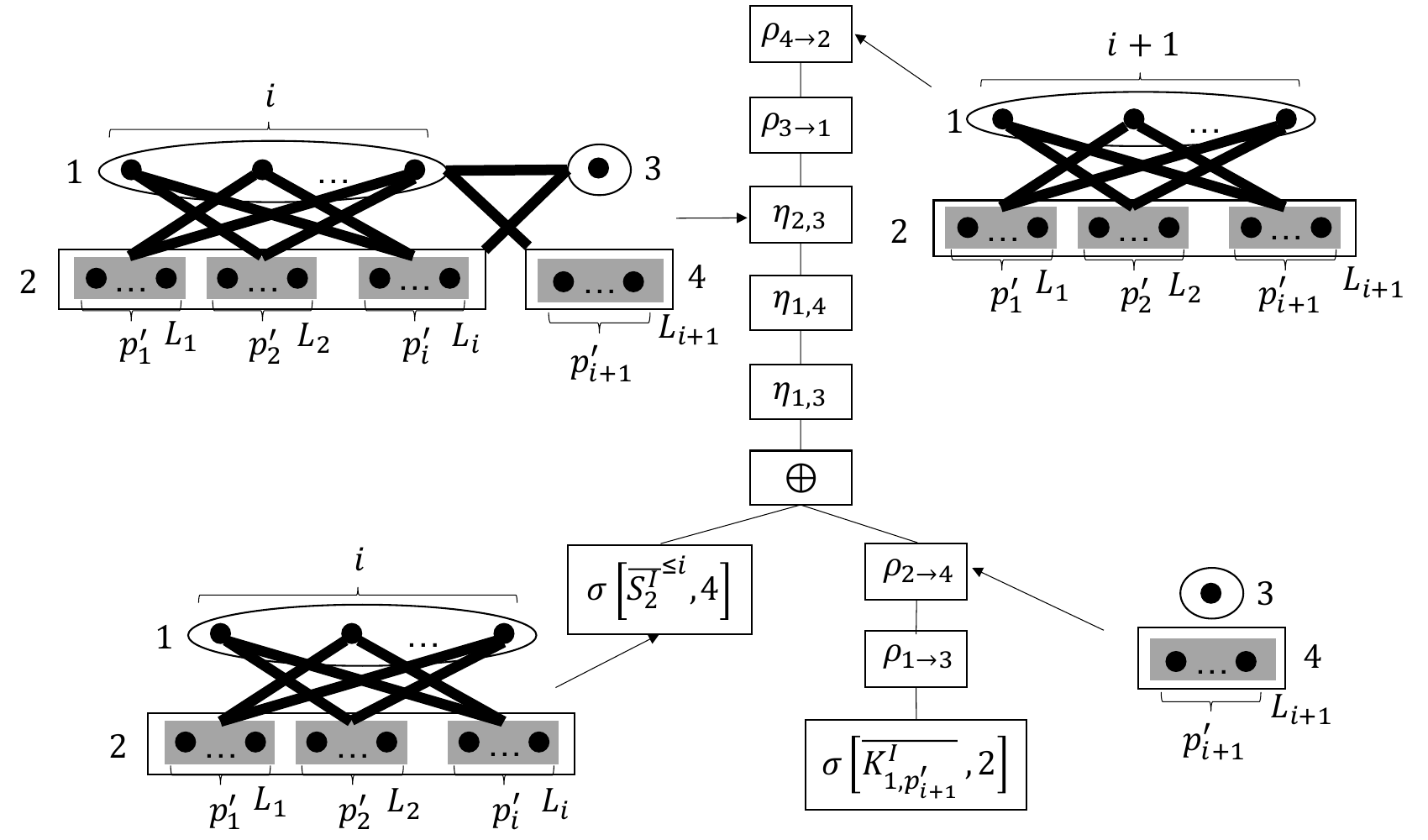}
    \caption{Expression tree for $\invCompl{S_2}$.}
    \label{fig:S2invCompl_tree}
\end{figure}

\end{document}